\documentclass[a4paper,11pt]{article}

\usepackage{graphicx}

\usepackage[T1]{fontenc}
\usepackage{amsmath}
\usepackage{amsfonts}
\usepackage{amsthm}
\usepackage{hyperref}
\usepackage{authblk}
\usepackage{geometry}

\newcommand{\ket}[1]{\ensuremath{\left|#1\right\rangle}} % Dirac Kets
\newcommand{\bra}[1]{\ensuremath{\left\langle#1\right|}} % Dirac Bras

\newtheorem{defi}{Definition}
\newtheorem{theo}{Theorem}
\newtheorem{lemm}{Lemma}

\title{Pseudo-telepathy games and genuine NS $k$-way nonlocality using graph states}
\author[2]{Anurag Anshu\thanks{a.anshu@iitg.ernet.in}} 
\author[1]{Mehdi Mhalla\thanks{mehdi.mhalla@imag.fr}\thanks{This work has been funded by the ANR-10-JCJC-0208 CausaQ grant.
}}
\affil[1]{CNRS LIG, University of Grenoble, France}
\affil[2]{I.I.T. Guwahati, India}
\date{}

\begin{document}
\maketitle
\sloppy
\begin{abstract}
We define a family of pseudo-telepathy games using graph states that extends the Mermin games. This family also contains a game used to define a quantum probability distribution that cannot be simulated by any number of nonlocal boxes. We extend this result, proving that the probability distribution obtained by the Paley graph state on 13 vertices (each vertex corresponds to a player) cannot be simulated by any number of 4-partite nonlocal boxes and that the Paley graph states on $k^{2}2^{2k-2}$ vertices provide a probability distribution that cannot be simulated by $k$-partite nonlocal boxes, for any $k$.
\end{abstract}

\section{Introduction}
 
Quantum nonlocality is one of those rare physical phenomena that are discovered to be deeply rooted in the foundations of physics long before they are properly understood and universally accepted. Originally used by Einstein, Podolosky and Rosen \cite{epr} in 1935 in their attempt to prove quantum mechanics incomplete, it was given a completely new avatar by John S. Bell in his seminal work \cite{bell} of 1964 and has now taken the form of a physical principle thanks to remarkable results like no-communication theorem \cite{peres} and provided interesting mathematical tools like nonlocal boxes \cite{popescu}.   

A nonlocal box refers to a virtual device that is \textit{non-signalling}, shared between multiple parties $(1,2\ldots n)$ and characterized by a joint probability distribution $P(a_1,a_2\ldots a_n|x_1,x_2\ldots x_n)$, where $(x_1,x_2\ldots x_n)$ is the `input' to the box and $(a_1,a_2\ldots a_n)$ is the `output'. The term \textit{nonlocal} implies that this probability cannot be written as $\sum_l p_l*P_l(a_1|x_1)*P_l(a_2|x_2)\ldots P_l(a_n|x_n)$ ($p_l>0,\sum_lp_l=1$), or equivalently that it cannot be simulated by a local classical protocol. \textit{Non-signalling} means that a set of players cannot acquire information about each other's input. In this paper we consider a strong version of \textit{non-signalling} to be satisfied by nonlocal boxes: even if a given set of players know each other's input and output, they cannot get information about the input of a player that is not in this set. This means that in the case where the inputs are in $\{0,1\}$ the probability distribution satisfies: for any input $(x_1,x_2\ldots x_n)$, output $(a_1,a_2\ldots a_n)$ and an index $i$, $\sum_{a_i} P(a_1,a_2\ldots a_n| x_1,\ldots x_{i-1},0,x_{i+1}\ldots x_n)=\sum_{a_i} P(a_1,a_2\ldots a_n| x_1,\ldots x_{i-1},1,x_{i+1}\ldots x_n)$. 

In bipartite case, the `PR Box', introduced in \cite{popescu}, is a nonlocal box that satisfies: $a_1 + a_2 = x_1.x_2 \bmod 2$. Its fundamental importance lies in the result that all the extremal points of bipartite non-signalling probability polytope are of this form\cite{barrett} and hence any non-signalling bipartite probability distribution can be simulated with PR boxes \cite{pironio}. PR boxes can also simulate many multipartite correlations as discussed in \cite{methot}. At present, only bipartite nonlocal boxes are well understood. A classification of extremal points of the nonlocal probability distribution polytope in tripartite scenario has been done recently in detail in \cite{scarani}, but a more intuitive understanding is still required. 
 
An important topic of interest in quantum information theory has been the characterization of quantum nonlocality. There have been two major approaches to this problem. The first is to consider the cost of simulating probability distribution exhibited by a physical system with nonlocal boxes \cite{gisin}, one way communication between observers \cite{toner}, bounded communication in the average or the worst case scenario \cite{branciardi}, etc.
 
The second approach is to look at the amount by which a nonlocal probability distribution violates Bell type inequalities such as the CHSH inequalities \cite{horne}. CHSH inequalities have been extended to the multipartite scenario by Svetlichny \cite{svetlinchny,seevinck} and weaker inequalities proposed by Pironio et. al. \cite{pironiomulti} that, within their notion of `k-way nonlocality', strengthen our intuitive understanding of multipartite nonlocal nature of certain quantum correlations. Another concept of `genuinely $k$-way NS nonlocal', as introduced in \cite{pironiomulti}, corresponds to those probability distributions that cannot be simulated by $k$-partite nonlocal boxes (nonlocal boxes that are shared between $k$ parties) and plays a central role in present article.  

We extend an important result in \cite{pironio} that uses only non-signalling to present a quantum correlation that cannot be simulated by PR Boxes. We show, without requiring any knowledge of the detailed nature of extremal points of nonlocal probability distribution polytope in multipartite setting, that there exist $(k^{2}2^{2k-2})$-partite quantum probability distributions obtained by Pauli measurements on graph states that cannot be simulated by $k$-partite nonlocal boxes, for any given $k$. This contrasts with Mermin-GHZ correlations \cite{methot} and nonlocal probability distributions of the form $P(a_1,a_2\ldots a_n|x_1,x_2\ldots x_n)=1/2^{n-1}$ if $a_1+a_2\ldots +a_n=f(x_1,x_2\ldots x_n)$ and 0 otherwise for any boolean function $f$ (as discussed in \cite{pironio}), both of which can be simulated by bipartite boxes (PR Boxes), for any $n$.

To achieve this goal, we consider the pseudo-telepathy games, an approach to understanding the nonlocal nature of quantum mechanics alongside Bell's inequalities \cite{bell}. Pseudo-telepathy games aim at providing a simple and natural interpretation of quantum nonlocality. They have been vividly described in Brassard et. al. \cite{broadbent} as protocols that can play important role in experimental verification of nonlocal nature of our world, in some cases even when measurement detectors are considerably inefficient \cite{hoyer}. On the theoretical side, they also provide an interesting measure of nonlocality, in terms of probability by which the best strategy of a classical player can win the games. 

We present a family of pseudo-telepathy games using graph states. If players share a graph state $\ket{G}$ then using a simple protocol consisting of measurements in the diagonal basis when the input is 1 and in the standard basis when the input is 0, they win perfectly the game defined on $G$. These games generalize a well-known Mermin's parity game, originally described as a 3-player game in \cite{greenberger} and studied as a general $n$-player game in \cite{brassard}. Moreover, the correlations considered in \cite{pironio} (and introduced in \cite{briegel}) that cannot be simulated using PR boxes, are the ones obtained in the game using graph $C_5$ (cycle on five vertices). 

Following the fact that $C_5$ is a special case of a more general family called Paley graphs, we show that the probability distribution obtained by the quantum strategy on Paley graph states on more than 5 vertices cannot be simulated by tripartite nonlocal boxes and that $4$-partite nonlocal boxes cannot simulate the probability distribution obtained in the game on the Paley graph state on 13 qubits. 
 
Finally, using a graph theoretic property called existential closure \cite{bonato1, erdos}, we infer that the probability distributions obtained with the Paley graph states on greater than $k^{2}2^{2k-2}$ qubits cannot be simulated by $k$-partite nonlocal boxes.     
  
\section{Pseudo-telepathy graph games}

We consider a game with $n$ players who are not allowed to communicate, each of them receives an input and is asked to provide an output. Given an input domain $I \subseteq I_1\times \ldots I_n$, and an output domain $O = O_1\times \ldots O_n$),
a game $\Gamma$ is characterized by a relation $ {\cal{L}}(\Gamma) \subseteq I \times O$ representing the set of loosing configurations \emph{i.e.}  if the players are asked questions $x_1,\ldots x_n$ and they answer $a_1,\ldots a_n$
with $x_1,\ldots x_n,a_1,\ldots a_n \in \  {\cal{L}} (\Gamma)$ the players lose the game.

The players win the game perfectly if losing configurations are never reached. Pseudo-telepathy games are games where using quantum mechanics, the players can win perfectly whereas classical players cannot.

Note that in general the set of legitimate questions $I$ is a subset of $I_1\times \ldots I_n$. When $I=I_1\times \ldots I_n$ we say that the game is without promise (it is easy to extend promise games to general games: the inputs not in $I$ have no losing condition associated to them). 
 
In the following, we will consider only the case where the inputs and outputs are in $\{0,1\}$.
  
Given a graph $G=(V,E)$, for any subset $D$ of vertices, we define its odd and even neighbourhood as follows: $Even(D)=\{v \in V, |N(v)\cap D|=0 \bmod 2\}$, and $Odd(D)=\{v \in V, |N(v)\cap D|=1 \bmod 2\}$, where $N(v)$ is the neighbourhood of $v$. Thus, $Odd(D)$ ($Even(D)$) is the set of vertices that have an odd (even) number of neighbours in $D$. Notions of odd and even neighbourhood have been very useful in the analysis of secret sharing schemes \cite{secret,secret2} and resource needed for preparation of a graph state \cite{hmp}, and they shall play a crucial role in the family of games we define below.

Before proceeding, note that given a graph $G=(V,E)$, a subset of vertices $D$ satisfies $D\subseteq Even(D)$ if and only if the subgraph induced by $D$ is Eulerian (all vertices have even degree).

Now we define graph games as follows. The players are identified with vertices of the graph. The players lose if and only if there exists a eulerian induced subgraph $D$, for which when 1 is asked to the players in $D$ and 0 to the players in $Odd(D)$, the sum of the answers modulo 2 of the players in $D\cup Odd(D)$ is different from the number of edges of the subgraph induced by $D$.

\begin{defi}\textbf{Graph game $\Gamma_G$:} Given a graph $G=(V,E)$ on $n$ vertices, each player is provided a question $x_i\in\{0,1\}$ and gives an answer $a_i \in \{0,1\}$  and the losing set of a graph game $\Gamma_G$  is defined as follows: $x_1,\ldots x_n, a_1,\ldots a_n$ is in the losing set ${\cal{L}}(\Gamma_G)$ if and only if $\exists D \subseteq V$ such that:
\begin{itemize}
\item  $D\subseteq Even(D)$ 
\item $x_k=1$ if $k\in D$ and 0 if $k\in Odd(D)$
\item  $\sum_{i\in D\cup Odd(D)} a_i= |E(D)| +1 \bmod 2$ where $E(D)$ is the set of edges of the subgraph induced by $D$.
\end{itemize}
\end{defi}
For any $D\subseteq Even(D)$, let $Q(D)$ be the set of questions for which $x_i=1$ if $i\in D$ and $x_i=0$ if $i\in Odd(D)$.

First, we provide a necessary and sufficient condition on a graph $G$, ensuring that if each player uses a classical deterministic strategy then the losing configuration of the game $\Gamma_G$ cannot be avoided. It may be recalled that a bipartite graph is a graph where all the edges are between two disjoint sets of vertices, and that a graph is bipartite if and only if it contains no odd cycle.

\begin{theo} \label{lmcnd} Given $G=(V,E)$, the graph game $\Gamma_G$ cannot be won perfectly by any classical deterministic strategy, if and only if $G$ is not bipartite. 
\end{theo}
\begin{proof}
Suppose $G$  is not bipartite, then it contains an odd cycle. By induction on the size of the smallest odd cycle, it is easy to see that $G$ contains an induced eulerian subgraph $C$.

Let $u_{1},u_{2}...u_{h}$ be vertices in $C$, with $h=|C|$. Suppose there exists a deterministic strategy that never loses for the game $\Gamma_G$, and let $a_{v}^{0}$ be the output of the player $v$ when his input is 0 and  $a_{v}^{1}$ if his input is $1$.

When $D=\{u_{i}\}$, $Odd(D)=N(u_i)$ and $D\subseteq Even(D)$. Thus if the players never lose then: 

\begin{equation}
a_{u_{i}}^{1} + \sum_{j\in N(u_{i})}a_{j}^{0} = 0 \bmod 2
\end{equation}

and when $D=C$

\begin{equation}
\sum_{l=1}^{h}a_{u_{l}}^{1} + \sum_{j\in Odd(U)}a_{j}^{0}=1   \bmod 2      
\end{equation}

Adding  equations 1 for $i$ from $1$ to $h$, we get $\sum_{i=1}^{h}a_{u_{i}}^{1} + \sum_{i=1}^{h}\sum_{j\in N(u_{i})}a_{j}^{0} = 0 \bmod 2$. In second term on left hand side, all $j$ that are in $Even(C)$ are added even number of times. Such $a_{j}^{0}$ thus do not appear in the equation. The result is: $ \sum_{l=1}^{h}a_{u_{l}}^{1} + \sum_{j\in Odd(C)}a_{j}^{0}=0   \bmod 2 $ which contradicts with equation $2$.

On the other hand, if $G$ is bipartite, then it contains no odd cycle. Thus any induced eulerian subgraph has an even number of edges, as it can be decomposed in cycles that are all even. Setting $a_i=0$ for all the players ensure that the game $\Gamma_G$ is won perfectly.
\end{proof}

Theorem \ref{lmcnd} directly implies (see \cite{broadbent}) that even with a probabilistic strategy and with shared random variables classical players have a non-zero probability to lose the game $\Gamma_G$ if $G$ is not bipartite.

However, for any graph game $\Gamma_G$, if the players share the graph state $\ket {G}$, there exists a strategy which ensures that they never lose.
Given a graph $G=(V,E)$, the graph state  $\ket {G}$ (\cite{gs}) is the quantum state that is common eigenvector with associated eigenvalue 1 of the Pauli operators $X_i\Pi_{j\in N(i)} Z_j$  for $i\in V$. 

\begin{theo} 
\label{thq}
There is a strategy that allows quantum players that share the state $\ket{G}$ to never lose in the  game $\Gamma_G$.
\end{theo}

\begin{proof} We show that if the players share the quantum state $\ket{G}$, and for input $1$($0$), they measure their qubit in the diagonal basis (standard basis), then the output completely avoids the losing conditions ${\cal{L(G)}}(\Gamma_G)$ on this game.  

Given a graph state, if  $(-1)^{s}\prod_{i\in V}\sigma_{i}^{h}$ is in the stabilizer of this graph state, where $\sigma_{i}^{h}$ are  Pauli matrices (and $h$ is an index unspecified here, that can take one of the four values corresponding to each pauli matrix), and the qubit $i$ is measured in the $\sigma_{i}^{h}$ basis, then it gives the measurement outcome $m_{i}\in{0,1}$ ($0$ corresponds to projector $(I+\sigma_{i}^{h})/2$ and $1$ to $(I-\sigma_{i}^{h})/2$). These measurement results  satisfy following relation: $\sum_{i} m_{i} = s \bmod 2$. This easily follows from observing that $(I+\sigma_{i}^{h})\sigma_{i}^{h} = I+\sigma_{i}^{h}$ and $(I-\sigma_{i}^{h})\sigma_{i}^{h} = -(I - \sigma_{i}^{h})$ and then applying this to the expression for probability of observing a measurement outcome $(m_1,m_2 \ldots m_{|V|})$ : $\bra{G}\prod_{i}(I+ (2m_i - 1)\sigma_{i}^{h})\ket{G}$. 

Now we show that, given any $D\in V$ such that $D\subseteq Even(D)$, there is an operator (which we construct below) in the stabilizer that corresponds to $s=|E(D)|$. Label the vertices in $V$ with integers $1,2 \ldots |V|$ such that vertices in $D$ are assigned the first $|D|$ integers. Then $\prod_{i\in D}X_{i}Z_{N(i)}$ is in the stabilizer and this is our desired operator. Since $XZX=-X$, every column containing odd number of $Z$ will contribute a $-1$. So total contribution when $-1$ from columns corresponding to $i\in D$ are multiplied will be $(-1)^{|E(D)|}$. Following previous paragraph, this concludes the proof.     
\end{proof}

Note that the quantum strategy for a game $\Gamma_G$ is the same as the one used in \cite{pironio}: when a player gets $1$ as input, he performs an $X$ measurement. When he gets $0$ as input, he performs a $Z$ measurement. The output is the classical measurement outcome.

As an example, we consider the game defined on the complete graph $\Gamma_{K_n}$. The graph state $\ket{K_n}$ is equivalent to the GHZ state on $n$ qubits up to local transformation. 

For any subset $D$ of vertices of $K_n$,  we have $D\subseteq Even(D)$ if and only if  $|D|=1\bmod 2$.  Now, if $|D|=1\bmod 2$ then the number of edges in the subgraph induced by $D$ is  $|E(D)|=|D|(|D|-1)/2 =(|D|-1)/2 \bmod 2$. Furthermore, $|D|$ is number of $1$s in the question. Hence, the losing conditions for $\Gamma_{K_n}$  exist when $ \sum_{j\in V}x_{j}= 1 \bmod 2$ and are of the form:

\begin{equation}
\sum_{i\in V} a_{i} = (\sum_{j\in V}x_{j}-1)/2 +1\bmod 2
\end{equation}  

This is a variation of the Mermin's parity games introduced in \cite{mermin}. These games are very interesting as they have small winning probability with classical strategy\cite{brassard} and their detector efficiency (see \cite{broadbent}) can be as low as $1/2$ and still distinguish between classical and quantum results. The games are defined as follows.

\textbf{Mermin's parity game:} This is a family of games for $n$ players, $n\geq 3$. The task that $n$ players face is the following: Each player $i$ receives as input a bit $x_i\in \left\{0,1\right\}$, which is also interpreted as an integer in binary, with the promise that $\sum_i x_i$ is divisible by $2$. The players must each output a single bit $a_i$ and the winning condition is:

\begin{equation}
\sum_{i=1}^{n} a_i = (\sum_{i=1}^{n}x_i)/2 \bmod 2 
\end{equation} 

$\Gamma_{K_n}$ can be transformed into Mermin's parity game using following local transformations by Player 1: $x_{1}\rightarrow x_{1}+1;a_{1}\rightarrow a_{1}+x_{1}+1$.  This means that given $x_{1}$, Player 1 applies above transformation on this question, plays Mermin's parity game and on the output $a_{1}$ of this game applies above transformation. Same transformation works for transforming Mermin's parity game into $\Gamma_{K_n}$. Thus optimal winning probability is same for both games. 

An important result regarding Mermin's game has been obtained in \cite{methot}. Theorem 9 of the paper shows that $n(n-1)/2$ PR boxes are sufficient for the simulation of $n$-partite Mermin's parity game. This and an another result discussed in \cite{pironio}, which shows that there is a simple probability distribution that cannot be simulated by any number of PR Boxes, motivates us to look at simulability of probability distributions using multipartite nonlocal boxes (PR Box is a bipartite nonlocal box). It forms the subject of next section.   

\section{Genuinely $k$-way NS nonlocality}

Given a graph over $n$ vertices, the quantum strategy that wins the corresponding graph game defined in Theorem \ref{thq} induces a probability distribution $P_{\ket{G}}(a_{1}a_{2}...a_{n}|x_{1}x_{2}...x_{n})$, where $x_{1}x_{2}...x_{n}$ are the questions and $a_{1}a_{2}...a_{n}$ are the answers. In the following, we study the simulability of these probability distributions using nonlocal boxes. We  identify various sufficient conditions on a graph $G$ that allow us to prove that the probability distribution $P_{\ket{G}}$ cannot be simulated by some multipartite nonlocal boxes and we exhibit a graph $G$ on 13 vertices such that $P_{\ket{G}}$ cannot be simulated by $4$-partite nonlocal boxes. The main result of this section is that the family of games defined using Paley graphs on $n$ vertices (which will be defined in next subsection) produces probability distributions that cannot be simulated by $o(\log n)$-partite nonlocal boxes. 

The quantum correlation in \cite{pironio} that cannot be simulated by any number of PR Boxes is same as that induced by quantum strategy for graph game on $C_5$ (cycle with $5$ vertices). We observe that since their argument is strongly based on non-signalling, it can be applied to a general multipartite setting without any prior knowledge of the probability distributions exhibited by multipartite nonlocal boxes. 

We adapt the following measure of nonlocality, called \textit{genuinely $k$-way NS nonlocal} and introduced in \cite{pironiomulti}: 

\begin{defi} A probability distribution is genuinely $k$-way NS nonlocal if it cannot be simulated by a protocol that involves sharing, between the players, nonlocal boxes that are at most $k$-partite. 
\end{defi}

Since the probability distribution obtained by quantum strategy on game $\Gamma_{C_5}$ cannot be simulated by bipartite nonlocal boxes, it is genuinely $2$-way NS nonlocal. In order to apply the argument in \cite{pironio} to obtain genuinely $k$-way NS nonlocal probability distributions, we shall consider a class of graphs defined below.

\begin{defi}
\label{defkod}
A graph $G=(V,E)$ is $k$-odd dominated ($k$-o.d.) if and only if for every subset $S\subseteq V$ with $|S|= k$, there exists a labeling of vertices in $S$ as $v_1,v_2,...v_k$ such that there exist $U_1,U_2...U_k$ satisfying: 
\begin{itemize}
\item $U_i\subseteq V\setminus S$
\item $U_i\subseteq Even(U_i)$
\item $Odd(U_i)\cap \left\{v_i,...v_k\right\}=v_i$
\end{itemize}
\end{defi}

The $k$-o.d. graphs satisfy the following property.

\begin{lemm}\label{kod} Given a graph $G=(V,E)$. If $G$ is $k$-o.d., then $G$ is $j$-o.d. for every $j<k$.
\end{lemm}

\begin{proof}
 Since $G$ is $k$-o.d., there exists a labeling of vertices in $S_1$ as $v_1,v_2...v_k$, with $v=v_l$ for some $l$ and  there exist $U_1,U_2...U_k \subseteq V\setminus S$ such that $U_i\subseteq Even(U_i)$ and $Odd(U_i)\cap \left\{v_i,...v_k\right\}=v_i$. Then a labeling on vertices in $S'$ can simply be defined to be $v_1...v_{l-1}v_{l+1}...v_k$. It is also easy to verify that $Odd(U_{i})\cap \left\{v_i,...v_{l-1}v_{l+1}...v_{k}\right\}=v_i$. This is true for every subset of size $k-1$. Hence $G$ is $(k-1)$-o.d.  
\end{proof}

To show that $P_{\ket{G}}$ is genuinely $k$-way NS nonlocal if graph $G$ is $k$-o.d, we follow the lines of the proof in \cite{pironio} for the non-simulability of $P_{\ket{C_5}}$ with PR boxes. Before proceeding, note that non-signalling allows to define the behavior of a nonlocal box when some of the players sharing it do not use it. 
Given a non-signalling probability distribution $P(a_1,a_2\ldots a_n|x_1,x_2\ldots x_n)$, $P(a_{l+1},a_{l+2}\ldots a_n|x_1,x_2\ldots x_n)=\sum_{a_1,a_2\ldots a_l} P(a_1,a_2\ldots a_n|x_1,x_2\ldots x_n) = P(a_{l+1},a_{l+2}\ldots a_n|x_{l+1},x_{l+2}\ldots x_n)$. Last term is the probability distribution over remaining players and is independent of input of first $l$ players. Thus remaining players have no information about whether the first $l$ players use the box or not.

General form of any protocol that uses nonlocal boxes and shared randomness among players to simulate a probability distribution can be given in following way \cite{pironio}. Let the nonlocal boxes that a particular player, for example $v_1$, shares with other players be labelled as $NL_1,NL_2\ldots NL_m$. Let his shared random variables be collectively represented by $\lambda$. Given an input $x_1$, $v_1$ puts a bit $y_1(x_1,\lambda)$ into $NL_1$. He gets an output $\alpha_1$ from the box. In $NL_2$, he inputs $y_2(\alpha_1,x_1,\lambda)$ and obtains $\alpha_2$. Continuing this way, he finally outputs $a_1(\lambda,\alpha_1,\alpha_2\ldots \alpha_m,x_1)$. Similar protocol is followed by every player.

\begin{lemm}
\label{nsgen}
Given a $k$-o.d. graph $G$, if the probability distribution $P_{\ket{G}}$ can be simulated by a protocol that uses a set of $j$-partite nonlocal boxes with $j\le k$ and shared randomness, then $P_{\ket{G}}$ can be simulated by an equivalent protocol in which any player who gets input $0$ does not use nonlocal boxes. 
\end{lemm}
\begin{proof}
Consider a protocol $\Pi$ and suppose that player $v_1$ gets input $x_1=0$. He uses a series of $m$ nonlocal boxes $NL_1,NL_2\ldots NL_m$ in this order, where input in any box is a function of outputs from previous boxes and shared random variables. Suppose $NL_m$ is shared by $j$ players, where $j\leq k$. Without loss of generality, we assume that these players are $v_1,v_2\ldots v_j$ and call this set $S$. Since graph is $j$-o.d (by lemma \ref{kod}), there exists a $U_1$ which satisfies $U_1\subseteq V\setminus S$, $U_1\subseteq Even(U_1)$ and $Odd(U_1)\cap \left\{v_1,...v_j\right\}=v_1$. 

For any question in $Q(U_1)$, we have $x_1=0$ and $\sum_{v\in U_1\cup Odd(U_1)}a_{v}=|E(U_1)| \bmod 2$. This condition for correlation on outputs of the players involves only vertices in $U_1\cup Odd(U_1)$  and does not involve any player in $S$ other than $v_1$. By non-signalling, even if the other players in $S$ do not use the box, the correlation does not change among players in $U_1\cup Odd(U_1)$. But if other pllayers do not use $NL_m$, the classical bit $v_1$ gets from $NL_m$ is uncorrelated with rest of the protocol. As a result, the final output of $v_1$ in $\Pi$ cannot depend on output he obtains from $NL_m$.  But if $v_1$ does not use the output of $NL_m$, he can very well terminate his protocol without inputting anything into $NL_m$.  

By same argument, if $v_1$ does not use $NL_m$, then he can terminate his protocol before using $NL_{m-1}$. Continuing this way, $v_1$ does not use any of the boxes if he gets $0$ as input. This argument holds for every player, proving the lemma.
\end{proof}

We thus have the following:

\begin{lemm} \label{odnl} Given a graph $G=(V,E)$, if $G$ is $k$-o.d. then $P_{\ket{G}}$ is genuinely $k$-way NS nonlocal.
\end{lemm}
\begin{proof}
Suppose a protocol $\Pi$ simulates $P_{\ket{G}}$ using nonlocal boxes that are at most $k$-partite. By Lemma \ref{nsgen}, a new protocol can be constructed in which a player, who gets input $0$, does not use any of the nonlocal boxes available to him. Now chose a player $v_1$ (without loss of generality) and consider a question belonging to $Q({v_1})$. For such a question, $v_1$ is input $1$ and rest players input $0$. No player other than $v_1$ uses the nonlocal boxes. As a result, the nonlocal boxes shared with $v_1$ act as an uncorrelated random variable for $v_1$. Thus $v_1$ does not use the output of nonlocal boxes, when given the input $1$. This holds for every player and hence $P_{\ket{G}}$ can be simulated by a protocol which involves only shared random variables. This is impossible, leading to a contradiction. So probability distribution is genuinely $k$-way NS nonlocal.  
\end{proof}   

We give below a sufficient condition under which a graph is $k$-o.d. This condition shall be useful in study of games over some graphs below (even though it is more restrictive, it is easier to check compared to $k$-o.d. property).  

\begin{lemm}\label{strcor} Given a graph $G=(V,E)$, If for every $T \subseteq V$ of size $j\leq k$ there exist disjoint subsets $W_1,W_2...W_{k-j+1}$ in $V\setminus T$ satisfying $W_i\subseteq Even(W_i)$ and $|Odd(W_i)\cap T|=1$, then $G$ is $k$-o.d.  
\end{lemm}
\begin{proof} For every subset $S\subseteq V$ of size $k$, it is sufficient to construct the labels $v_1,v_2...v_k$ and subsets $U_1,U_2...U_k$ as stated in Definition \ref{defkod}. Given such a $S$, consider the case when $T=S$. By the condition in the lemma, there exists a $W_1 \subseteq V\setminus T$ that satisfies $|Odd(W_1)\cap T|=1$ and $W_1\subseteq Even(W_1)$. Label the vertex in $S$ adjacent to $W_1$ as $v_1$ and chose $U_1$ to be the subset $W_1$. Now, consider the case when $T=S\setminus v_1$. There exist two disjoint subsets $W_2,W'_2$ of $V\setminus T$ that satisfy $|Odd(W_2)\cap T|=|Odd(W'_2)\cap T|=1$ and $W_2\subseteq Even(W_2),W'_2\subseteq Even(W'_2)$. If $W_2$ contains $v_1$, then label the vertex in $S$ adjacent to $W'_2$ as $v_2$ and chose $U_2$ to be the subset $W'_2$. Else label the vertex adjacent to $W_2$ as $v_2$ and chose $U_2$ to be the subset $W_2$. Continuing this way, the vertices in $S$ can be labelled as $v_1,v_2...v_k$ and the corresponding sets $U_1,U_2...U_k$ can be constructed. Hence $G$ is $k$-o.d.   
\end{proof}

\subsection{Example of $3-$o.d. and $4-$o.d. graphs}

The Paley graphs on $n$ vertices ($Pal_n$) are defined when $n$ is a prime power and satisfies $n=1\bmod 4$. Vertices are labelled with numbers $0,1...(n-1)$ and two vertices $a,b$ are adjacent if and only if  $a-b=m^2 \bmod n$ for some $m$. These graphs have recently been investigated for their interesting properties under local complementation \cite{local}. They are up to our knowledge the only known family of graphs with minimal degree by local complementation larger than the square root of their order, which implies by \cite{hmp} that to prepare the Paley graph states by only measurement, one needs $\sqrt{n}$ qubits measurements. They are also the best known family of graphs for secret sharing with graph states and thus for some quantum codes \cite{secret,secret2}.

Furthermore, these graphs have interesting properties such as: they are strongly regular ($Srg(n,(n-1)/2,(n-5)/4,(n-1)/4)$, see for example \cite{bonato}), self-complementary, edge transitive (any edge is same as any other edge, up to a relabeling of vertices) and vertex transitive (any vertex is equivalent to any other vertex, up to a relabeling of vertices).  

Using Lemma \ref{strcor} we show that every Paley graph having more than $5$ vertices is $3$-o.d. and that the probability distribution obtained by the game on $\ket{Pal_{13}}$ is genuinely $4$-way NS nonlocal.

\begin{lemm} \label{3od}
 For every $n\ge 9$, $Pal_n$ is $3-$o.d.
\end{lemm}
\begin{proof}
 $Pal_n$ is strongly regular: every vertex has $(n-1)/2$ neighbours, any two adjacent vertices have $(n-5)/4$ vertices in their common neighbourhood and any two non-adjacent vertices have $(n-1)/4$ vertices in their common neighbourhood. Then given any two vertices, there are at least $(n-1)/2 + (n-1)/2 - 2 - (n-5)/4  = 3(n-1)/4 -1$ vertices in remaining graph adjacent to exactly one of these vertices. For any subset $S=\{v_1,v_2,v_3\}$, let $J$ be the set vertices adjacent to exactly one of $v_1,v_2$. Number of vertices in $J$ are at least $3(n-1)/4 - 1$.  If $v_3$ is not in $J$, then number of vertices in $J$ adjacent to $v_3$ is at most $(n-1)/2$. Hence, there are at least $(n-1)/4 -1$ vertices adjacent to exactly one of the vertices in $S$. If $v_3$ is in $J$, then number of vertices in $J \setminus v_3$ is at least $3*(n-1)/4 -2$. But the number of neighbours of $v_3$ in $J$ is at most $(n-1)/2-1$. Hence, there are at least $(n-1)/4 -1$ vertices adjacent to exactly one of the vertices in $S$, which is greater than $1$ for $n\ge 9$. Proof follows from Lemma \ref{strcor}, with each disjoint subset $W_i$ being a vertex. 
 \end{proof}

 Note that the argument in above lemma cannot be extended to subsets of size $4$.
 
 \begin{theo}  
The probability distribution   $P_{\ket{Pal_{13}}}$ (Figure 1) is genuinely $4$-way NS nonlocal.  
\end{theo}
\begin{figure}
	\centering
		\includegraphics[scale=0.7]{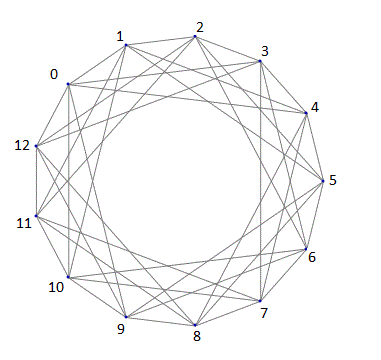}
		\caption{Paley graph on 13 vertices}
	\label{fig:figure 1}
\end{figure}
\begin{proof}

Using Lemma \ref{strcor}, it is sufficient to prove that for every $j\leq 4$, if for every subset $S'\subseteq V$ of size $j$ there exist $5-j$ disjoint subsets $U_1,U_2...U_{5-j}$ in $V\setminus S'$ each of which satisfy $|Odd(U_i)\cap S'|=1$ and $U_i\subseteq Even(U_i)$, then $G$ is $4$-o.d.

By Lemma \ref{3od}, Lemma \ref{strcor} is satisfied for $j=1,2,3$ with $U_i$ of size $1$. 

For $j=4$, we verify it case by case below, whence $U_i$ may not be subsets of size $1$ in some cases, but they satisfy $U_i\subseteq Even(U_i)$. 

By edge transitivity and self-complementarity, we can always select $1,2$ or $1,3$ as two elements of set $S$. Furthermore, a subgraph isomorphic to $K_{4}$ (complete graph on 4 vertices) does not exist in $Pal_{13}$. So there are no set of four vertices with no edge between them, as the graph is self-complementary \cite{sachs}. Our cases are as follows, classified on the basis of number of edges. 

\begin{enumerate}

\item \textbf{Five edges:} Possibilities are: $1,2,5,6$ and $1,2,5,11$. 
\begin{itemize}
\item For case $1$- $7$ is adjacent to only $6$. 
\item For case $2$- $9$ is adjacent to $5$.
\end{itemize}

\item \textbf{Four edges:} We have following possibilities: $1,2,3,4$; $1,2,6,10$ (rectangles); $1,2,11,0$; $1,2,3,11$; $1,2,4,11$; $1,2,6,11$; $1,2,7,11$; $1,2,8,11$ (triangle and an edge). 

\begin{itemize}
\item Case $1$- $10$ is adjacent to only $1$.
\item Case $2$- $12$ is adjacent to only $1$. 
\item Case $3$- $8$ is adjacent to only $11$. 
\item Case $4$- $8$ is adjacent to only $11$.  
\item Case $5$- $6$ is adjacent to $2$ only.  
\item Case $6$- $0$ is adjacent only to $1$.
\item Case $7$- $0$ is adjacent only to $1$.  
\item Case $8$- $0$ is adjacent only to $1$.
\end{itemize}

\item \textbf{Three edges:} Possibilities are: $1,2,9,11$ (triangle with an independent vertex); $1,2,4,10$ (star); $1,2,3,7$; $1,2,3,12$; $1,2,6,7$; $1,2,6,9$; $1,2,12,8$; $1,2,12,9$; $1,3,4,12$; $1,3,4,6$; $1,3,6,10$; $1,3,7,10$. 
\begin{itemize}
\item Case $1$- $7$ is adjacent to only $11$.
\item Case $2$- $9$ is adjacent to only $10$. 
\item Case $3$- $8$ is adjacent only to $7$.  
\item Case $4$- $9$ is adjacent to only $12$.
\item Case $5$- $0$ is adjacent to $1$ only.  
\item Case $6$- $8$ is adjacent to only $9$.
\item Case $7$- $7$ is adjacent to $8$ only. 
\item Case $8$- $4$ is adjacent to $1$ only. 
\item Case $9$- $10$ is adjacent to $1$ only. 
\item Case $10$- $11$ is adjacent to $1$ only. 
\item Case $11$- $12$ is adjacent to $3$ only. 
\item Case $12$- $9$ are adjacent to only $10$.
\end{itemize}

\item \textbf{Two edges:} Possibilities are: $1,2,4,9$; $1,3,8,11$; $1,3,5,8$ (two edges forming a ray and one vertex independent); $1,2,7,8$; $1,3,6,11$; $1,3,5,7$ (one edge on two vertices and another edge on remaining vertices).
\begin{itemize} 
\item Case $1$- $7$ is adjacent to only $4$. 
\item Case $2$- $6$ is adjacent to $3$ only.  
\item Case $3$- $10$ is adjacent to $1$ only.
\item Case $4$- $9$ is adjacent only to $1$. 
\item Case $5$- $9$ is adjacent to only $6$.
\item Case $6$- $12$ is adjacent to only $3$.
\end{itemize}

\item \textbf{One edge:} Possibilities are: $1,2,9,7$; $1,3,9,11$; $1,3,8,10$. These cases have the difficulty that there is no vertex in remaining graph that is adjacent to only one of the vertices.   
\begin{itemize}
\item Case $1$- We consider $U=\{4,6,12\}$. $Odd(U)$ contains only $1$ among $1,2,7,9$. 

\item Case $2$- We consider $U=\{10,8\}$. $Odd(U)$ contains only $1$.

\item Case $3$- We consider $U=\{7,9\}$. $Odd(U)$ contains only $3$. 
\end{itemize}

\item \textbf{No edge:} There are no possibilities here. 
\end{enumerate}
\end{proof}

\subsection{$k$-existential closure and $k$-odd domination}

An another sufficient condition can be given for $k$-o.d., which is motivated by a more restrictive graph theoretic notion called $k$-existential closure, originally introduced by Erdos, Renyi in \cite{erdos} and defined as follows:

\begin{defi} A graph $G$ is $k$-existentially closed or $k$-e.c. if for every $k$-element subset $S$ of the vertices, and for every subset $T$ of $S$, there is a vertex not in $S$ which is joined to every vertex in $T$ and to no vertex in $S\setminus T$. 
\end{defi}

\begin{lemm} \label{ecod} If a graph $G$ is $k$-e.c, then it is $k$-o.d.
\end{lemm}
\begin{proof} Consider any subset $S$ of size $k$. Assign any labeling to vertices and let it be $v_1,v_2...v_k$. Consider the subset $\{v_i,...v_k\}$ for all $i\leq k$. Let $T=\{v_1,v_2...v_i\}$. Then there exists a vertex $u_i$ in $V\setminus S$ that is adjacent to all vertices in $T$ and to no vertex in $S\setminus T$. Further, $u_i\in Even(u_i)$. So  $U_i=\{u_i\}$ satisfies the condition $Odd(U_i)\cap \left\{v_i,...v_k\right\}=v_i$. Hence the graph is $k$-o.d. 
\end{proof}

This condition allows us to sketch the nature of nonlocality in many quantum probability distributions. It is known that almost all finite graphs are $k$-e.c for some $k\ge 2$\cite{bonato1} and hence, by above lemma, quantum probability distributions corresponding to almost all graph games are genuinely $k$-way NS nonlocal, for some $k\ge 2$. Furthermore, for Paley graph states, we have the following general result:

\begin{theo} For any $k$, the quantum probability distribution obtained with Paley graph states of size greater than $k^{2}*2^{2k-2}$ is genuinely $k$-way NS nonlocal.  
\end{theo}   

\begin{proof}
As proved in \cite{blass}, a Paley graph having more than $k^{2}2^{2k-2}$ vertices is $k$-e.c. Thus by Lemma \ref{ecod} it is also $k$-o.d. Hence, Lemma \ref{odnl} allows to conclude.
\end{proof}

Note that the graphs $Pal_n$ are $3$-e.c. only for $n\ge 29$ (see \cite{bonato}), thus the general result could not be applied for our discussion related to $Pal_{13}$ in previous subsection. A detailed discussion on  known families of $k$-e.c. graphs, that include variants on Paley graphs, graphs related to Hadamard matrices etc., can be found in \cite{bonato1}.

\section{Acknowledgements}

The authors would like to thank Simon Perdrix and Peter H{\o}yer for fruitful discussions and the anonymous referees for improvements on the paper.
%This work has been funded by the ANR-10-JCJC-0208 CausaQ grant.

\section{Conclusion}

Using graph states and simple measurements, we have defined a new family of pseudo-telepathy games generalizing the Mermin game, and we have shown that the probability distribution obtained using the Paley graph states exhibits a strong multipartite nonlocality. 

It would be interesting to study more thoroughly the graph games defined here, for example to find the general expression for probability of winning the games in the best classical strategy. This question shall allow us to study the behaviour of these games with reference to inefficiency in detectors, a problem addressed in general in \cite{broadbent}. 
 
The results in \cite{hmp} show that preparing a graph state $\ket{G}$ on $n$ qubits with measurement only, requires measurements on $\delta_{loc}(G)+1$ qubits simultaneously, where $\delta_{loc}(G)$ is the minimum degree by local complementation. Thus for graph states, the minimum degree by local complementation is a measure of multipartite nonlocality for which the complete graph (the GHZ state) behave poorly ($\delta_{loc}(K_n)=1$) and that is high for Paley graph states \cite{local} ($\delta_{loc} (Pal_n) \ge \sqrt{n}$) (it is also mentioned in \cite{local} that the question of the existence of a subfamily of Paley graph states requiring measurements on $c.n$ qubits for some constant $c$ is equivalent to a known conjecture in code theory). Thus relating $k$-o.d. or a weaker condition for genuine $k$-way NS nonlocality with the minimum degree by local complementation could give better bounds for Paley graphs (in this paper we proved  only genuine ($\log k$)-way NS nonlocality for Paley graph states).
It would also be interesting to have a combinatorial necessary condition for genuine NS $k$-way nonlocality for the probability distributions obtained by the graph games.

Another direction is the extension where a player can have a set of qubits (a set of vertices of the graph) both in the case where the input is binary and when the input is non-binary. An interesting result regarding simulability when input can be non-binary has already be obtained in \cite{methot} in which the simulability of the magic square game (that involves 3 possible inputs) using PR boxes has been studied.

\end{document}